\documentclass[10pt, conference, compsocconf]{IEEEtran}

\def\IEEEbibitemsep{0pt plus .5pt}
\usepackage{amsmath,amssymb,amsfonts}
\usepackage{graphicx}
\usepackage{textcomp}
\usepackage{xcolor}

\usepackage{setspace}

\usepackage{booktabs} 
\usepackage{array}
\usepackage{multirow}
\usepackage{url}
\usepackage{caption}
\usepackage[caption=false,font=footnotesize]{subfig}
\usepackage{cleveref}
\usepackage{MnSymbol}
\usepackage[normalem]{ulem}
\usepackage[utf8]{inputenc}
\usepackage[T1]{fontenc}
\usepackage[english]{babel}
\usepackage{comment}

\usepackage{amsthm,amsmath,bm}
\usepackage[utf8]{inputenc}

\def\BibTeX{{\rm B\kern-.05em{\sc i\kern-.025em b}\kern-.08em
    T\kern-.1667em\lower.7ex\hbox{E}\kern-.125emX}}

\usepackage[ruled]{algorithm2e} 
\SetAlFnt{\small}
\SetAlCapFnt{\small}
\SetAlCapNameFnt{\small}
\SetAlCapHSkip{0pt}
\IncMargin{-\parindent}
\SetKwComment{Comment}{$\triangleright$\ }{}

\SetCommentSty{mycommfont}

\newtheorem{mythm}{Theorem}

\newcommand{\always}{\square}
\newcommand{\eventually}{\lozenge}
\newcommand{\until}{\mathcal{U}_I}

\newcommand{\ew}{\boxbox_{\ra}}
\newcommand{\sw}{\diamonddiamond_{\ra}}
\newcommand{\ag}{\mathcal{A}_{\ra}}
\newcommand{\ct}{\mathcal{C}_{\ra}}

\newcommand{\op}{\mathrm{op}}

\newcommand{\avg}{\mathrm{avg}}

\newcommand{\nb}{L_{\ra}}
\newcommand{\nbx}{\alpha_{\ra}^x(\omega, t, l)}

\newcommand{\sat}{\models}
\newcommand{\eqdef}{\Leftrightarrow}

\newcommand{\everywhere}{\boxbox}
\newcommand{\somewhere}{\diamonddiamond}
\newcommand{\agr}{\mathcal{A}}

\newcommand{\ra}{\mathcal{D}}

\newcommand{\tablefontsize}{\scriptsize}

\definecolor{brilliantlavender}{rgb}{0.96, 0.73, 1.0}
\definecolor{blond}{rgb}{0.98, 0.94, 0.75}
\definecolor{celadon}{rgb}{0.67, 0.88, 0.69}
\definecolor{columbiablue}{rgb}{0.61, 0.87, 1.0}
\definecolor{lavenderblush}{rgb}{1.0, 0.94, 0.96}
\definecolor{electriclavender}{rgb}{0.96, 0.73, 1.0}

\newcommand{\ent}[1]{\colorbox{lavenderblush}{#1}}
\newcommand{\temporal}[1]{\colorbox{celadon}{#1}}
\newcommand{\aggregation}[1]{\colorbox{blond}{#1}}
\newcommand{\condition}[1]{\colorbox{pink}{#1}}
\newcommand{\comparison}[1]{\colorbox{columbiablue}{#1}}
\newcommand{\spatial}[1]{\colorbox{electriclavender}{#1}}

\newcolumntype{L}[1]{>{\raggedright\let\newline\\\arraybackslash\hspace{0pt}}m{#1}}
\newcolumntype{C}[1]{>{\centering\let\newline\\\arraybackslash\hspace{0pt}}m{#1}}
\newcolumntype{R}[1]{>{\raggedleft\let\newline\\\arraybackslash\hspace{0pt}}m{#1}}

\newcommand{\sectref}[1]{Section~\ref{#1}}
\newcommand{\figref}[1]{Figure~\ref{#1}}
\newcommand{\tabref}[1]{Table~\ref{#1}}

\newcommand{\algref}[1]{Algorithm~\ref{#1}}

\newtheorem{example}{Example}
\newtheorem{lemma}{Lemma}

\graphicspath{{Figure/}}

\begin{document}

\title{SaSTL: Spatial Aggregation Signal Temporal Logic for\\ Runtime Monitoring in Smart Cities\\
}

\author{\IEEEauthorblockN{Meiyi Ma}
\IEEEauthorblockA{Department of Computer Science\\
University of Virginia\\
Charlottesville, USA\\
Email: meiyi@virginia.edu}
\and
\IEEEauthorblockN{Ezio Bartocci}
\IEEEauthorblockA{Faculty of Informatics\\
TU Wien\\
Vienna, Austria\\
Email: ezio.bartocci@tuwien.ac.at}
\and
\IEEEauthorblockN{Eli Lifland, John Stankovic, Lu Feng}
\IEEEauthorblockA{Department of Computer Science\\
University of Virginia\\
Charlottesville, USA\\
Email: \{edl9cy, stankovic, lu.feng\}@virginia.edu}
}


\maketitle

\begin{abstract}
We present SaSTL---a novel Spatial Aggregation Signal Temporal Logic---for the efficient runtime monitoring of safety and performance requirements in smart cities.
We first describe a study of over 1,000 smart city requirements, some of which can not be specified using existing logic such as Signal Temporal Logic (STL) and its variants. 
To tackle this limitation, we develop two new logical operators in SaSTL to augment STL for expressing spatial aggregation and spatial counting characteristics that are commonly found in real city requirements. 
We also develop efficient monitoring algorithms that can check a SaSTL requirement in parallel over multiple data streams (e.g., generated by multiple sensors distributed spatially in a city).
We evaluate our SaSTL monitor by applying to two case studies with large-scale real city sensing data (e.g., up to 10,000 sensors in one requirement). The results show that SaSTL has a much higher coverage expressiveness than other spatial-temporal logics, and with a significant reduction of computation time for monitoring requirements. 
We also demonstrate that the SaSTL monitor can help improve the safety and performance of smart cities via simulated experiments.

\end{abstract}



\section{Introduction}
\label{sect:intro}
Smart cities are emerging around the world. Examples include Chicago's Array of Things project~\cite{arrayofthings}, IBM's Rio de Janeiro Operations Center~\cite{rio2012center} and Cisco's Smart+Connected Operations Center~\cite{cisco-center}, just to name a few. Smart cities utilize a vast amount of data and smart services to enhance the safety, efficiency, and performance of city operations~\cite{ma2019data,yuan2018dynamic}. 
There is a need for monitoring city states in real-time to ensure safety and performance requirements~\cite{ma2017cityguard}.
If a requirement violation is detected by the monitor, the city operators and smart service providers can take actions to change the states, such as improving traffic performance, rejecting unsafe actions, sending alarms to polices, etc.
The key \textbf{challenges} of developing such a monitor include how to use an expressive, machine-understandable language to specify smart city requirements, 
and how to efficiently monitor requirements that may involve multiple sensor data streams (e.g., some requirements are concerned about thousands of sensors in a smart city).

Previous works~\cite{zhang2018detecting, sheng2019case, ma2017runtime,haghighi2015spatel} have proposed solutions to monitor smart cities using formal specification languages and their monitoring machinery.
One of the latest works, CityResolver \cite{ma2018cityresolver} uses Signal Temporal Logic (STL)~\cite{Maler2004} to support the specification-based monitoring of safety and performance requirements of smart cities. 
However, STL is not expressive enough to specify smart city requirements concerning \emph{spatial} information such as \emph{``the average noise level within 1 km of all elementary schools should always be less than 50 dB''}. 
There are some existing spatial extensions of STL (e.g.,  SSTL~\cite{NenziBCLM15}, SpaTeL~\cite{haghighi2015spatel} and 
STREL~\cite{BartocciBLN17}), which can express requirements 
such as \emph{``there should be no traffic congestion on 
all the roads in the northeast direction''}. 
But they are not expressive enough to specify requirements like \emph{``there should be no traffic congestion on all the roads on average''}, or \emph{``on 90\% of the roads''}, which require the aggregation and counting of signals in the spatial domain. To tackle these challenges and limitations, we develop a novel Spatial Aggregation Signal Temporal Logic (SaSTL), which extend STL with two new logical operators for expressing spatial aggregation and spatial counting characteristics that are commonly found in real city requirements. 
More specifically, this paper has the following major \textbf{contributions}:

\begin{figure*}[!t]
    \centering
    
    \includegraphics[width = 0.85\textwidth]{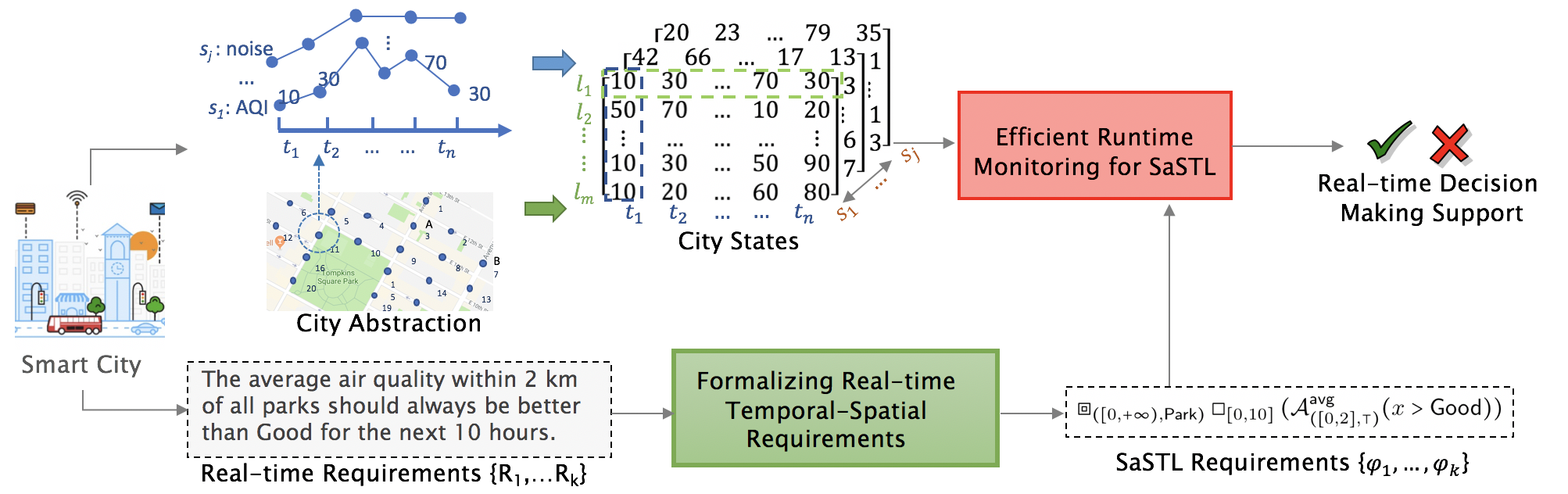}
    \caption{{A framework for runtime monitoring of real-time city requirements}}
    \label{fig:overview}
    \vspace{-1.3em}
\end{figure*}

(1) To the best of our knowledge, this is the first work studying and annotating over 1,000 real smart city requirements across different domains to identify the gap of expressing smart city requirements with existing formal specification languages. 
As a result, we found that aggregation and counting signals in the spatial domain (e.g., for representing sensor signals distributed spatially in a smart city) are extremely important for specifying and monitoring city requirements. 

(2) Drawing on the insights from our requirement study, we develop a new specification language SaSTL, which extends STL with a \emph{spatial aggregation} operator and a \emph{spatial counting} operator.
SaSTL can be used to specify Point of Interests (PoIs), the physical distance, spatial relations of the PoIs and sensors, aggregation of signals over locations, degree/percentage of satisfaction and the temporal elements in a very flexible spatial-temporal scale. 

(3) We compare SaSTL with some existing specification languages and show that SaSTL has a much higher coverage experessiveness (95\%) than STL (18.4\%), SSTL (43.1\%) or STREL (43.1\%) over 1,000 real city requirements.

(4) We develop novel and efficient monitoring algorithms for SaSTL. In particular, we present two new methods to speed up the monitoring performance: (i) dynamically prioritizing the monitoring based on cost functions assigned to nodes of the syntax tree, and (ii) parallelizing the monitoring of spatial operators among multiple locations and/or sensors. 

(5) We evaluate the SaSTL monitor by applying it to a case study of monitoring real city data collected from the Chicago's Array of Things~\cite{arrayofthings}. The results show that SaSTL monitor has the potential to help identify safety violations and support the city managers and citizens to make decisions. 

(6) We also evaluate the SaSTL monitor on a second case study of conflict detection and resolution among smart services in simulated New York city with large-scale real sensing data (e.g., up to 10,000 sensors in one requirement). 
Results of our simulated experiments show that SaSTL monitor can help improve the city’s performance (e.g.,  21.1\%  on  the  environment and 16.6\% on public safety), with a significant reduction of  computation  time  compared  with  previous  approaches.

\section{Approach Overview}
\label{sect:framework}

\figref{fig:overview} shows an overview of our SaSTL runtime monitoring framework for smart cities. 
We envision that such a framework would operate in a smart city's central control center (e.g., IBM's Rio de Janeiro Operations Center~\cite{rio2012center} or Cisco's Smart+Connected Operations Center~\cite{cisco-center}) where sensor data about city states across various locations are available in real time. 
The framework would monitor city states and check them against a set of smart city requirements at runtime. The monitoring results would be presented to city managers to support decision making. 
The framework makes abstractions of city states in the following way. 
The framework formalizes a set of smart city requirements (See~\sectref{sect:moti}) to some machine checkable SaSTL formulas (See~\sectref{sect:spec}). 
Different data streams (e.g. CO emission, noise level) over temporal and spatial domains can be viewed as a 3-dimensional matrix. For any signal $s_j$ in signal domain $S$, each row is a time-series data at one location and each column is a set of data streams from all locations at one time. Next, the efficient real-time monitoring for SaSTL
verifies the states with the requirements and outputs the Boolean satisfaction to the decision makers, who would take actions to resolve the violation. To support decision making in real time, we improve the efficiency of the monitoring algorithm in \sectref{sect:alg}.  
 We will describe more details of the framework in the following sections.

\section{Analysis of Real City Requirements}
\label{sect:moti}

To better understand real city requirements, we conduct a requirement study. 
We collect and statistically analyze 1000 quantitatively specified city requirements 
(e.g., standards, regulations, city codes, and laws) across different application domains, including energy, environment, transportation, emergency, and public safety from large cities (e.g. New York City, San Francisco, Chicago, Washington D.C., Beijing, etc.). Some examples of these city requirements are highlighted in \tabref{tab:reqexamples}. We identify key required features to have in a specification language and its associated use in a city runtime monitor. The summarized statistical results of the study and key elements we identified (i.e., temporal, spatial, aggregation, entity, comparison, and condition) are shown in \tabref{tab:elements}.

\begin{table*}[!h]
	\caption{Examples of city requirements from different domains (The key elements of a requirement are highlighted as, \temporal{temporal}, \spatial{spatial}, \aggregation{aggregation}, \ent{entity}, \condition{condition}, \comparison{comparison}.  )}
	\centering
	\tablefontsize
	\label{tab:reqexamples}
	\begin{tabular}{|L{1.5cm}|L{15cm}|}
		\hline
		\textbf{Domain} & \textbf{Example}\\\hline 
		
		\multirow{3}{*}{\textbf{Transportation}} & \comparison{Limits} \ent{vehicle idling} to \temporal{one minute} adjacent to \aggregation{any} \spatial{school, pre-K to 12th grade}, public or private, in the \spatial{City of New York}~\cite{r1}. \\\cline{2-2}
		
		& The engine, power and exhaust mechanism of each motor vehicle shall be equipped, adjusted and operated to \comparison{prevent} the escape of a trail of \ent{visible fumes or smoke} for \comparison{more than} \temporal{ten (10) consecutive seconds}~\cite{r10}.\\\cline{2-2}

		& \comparison{Prohibit} \ent{sight-seeing buses} from using \aggregation{all} \spatial{bus lanes} between the hours of \temporal{7:00 a.m. and 10:00 a.m.} on \temporal{weekdays}~\cite{r2}. \\\hline

		\multirow{2}{*}{\textbf{Energy}} & Operate the \ent{system} to \comparison{maintain} \spatial{zone} temperatures \comparison{down} to 55°F or \comparison{up to} 85°F~\cite{r3}.\\\cline{2-2}
		
		& The \aggregation{total} \ent{leakage} shall be \comparison{less than or equal} to 4 cubic feet \aggregation{per} \temporal{minute} \aggregation{per} \spatial{100 square feet} of \spatial{conditioned floor area}~\cite{r4}. \\\hline
		
		\multirow{2}{*}{\textbf{Environment}}  
		
		& LA Sec 111.03 \comparison{minimum} \ent{ambient noise level} table: \spatial{ZONE M2 and M3} -- \temporal{DAY}: 65 dB(A) \temporal{NIGHT}: 65 dB(A)~\cite{r11}. \\\cline{2-2}
		
		 & The \aggregation{total amount} of \ent{HCHO emission} should be \comparison{less than} 0.1mg \aggregation{per} m$^3$ \temporal{within an hour}, and the \aggregation{total amount} of PM10 emission should be \comparison{less than} 0.15 mg \aggregation{per} m$^3$ \temporal{within 24 hours}~\cite{r6}. \\\hline
	\multirow{2}{*}{\textbf{Emergency}}  & NYC Authorized \ent{emergency vehicles} may \condition{disregard} 4 primary rules regarding traffic~\cite{r7}.  \\\cline{2-2}
	 & \comparison{At least} one \ent{ambulance} should be equipped \aggregation{per} 30,000 population (counted \spatial{by area}) to obtain the shortest radius and fastest response time~\cite{r8}. \\\hline
	\textbf{Public Safety} &  \ent{Security staff} shall visit \comparison{at least once} \temporal{per week} in \spatial{public schools}~\cite{r9}. \\\hline
	\end{tabular}
	\vspace{-1.4em}
\end{table*}

\begin{table}[!h]
	\caption{Key elements of city requirements and statistical results from 1000 real city requirements}
	\centering
	\label{tab:elements}
	\tablefontsize
\begin{tabular}{|L{1.2cm}|l|c|L{3.8cm}|}
\hline
\textbf{Element }                                                                & \textbf{Form}             & \textbf{Num} & \textbf{Example}                                           \\ \hline
\multirow{4}{*}{\begin{tabular}[c]{@{}l@{}}\textbf{Temporal}\\ \end{tabular}}    & Dynamic Deadline &  77   & limit ... to one minute                           \\ \cline{2-4} 
                                                                            & Static Deadline  &   98  & at least once a week                              \\ \cline{2-4} 
                                                                            & Interval         &  168   & from 8am to 10am; within 24 hours;                \\ \cline{2-4} 
                                               
                                                                            & Default          &  657   & The noise (always) should not exceed 50dB.        \\ \hline
\multirow{3}{*}{\begin{tabular}[c]{@{}l@{}}\textbf{Spatial}\\ \end{tabular}}     & \textbf{PoIs/Tags}             &  \textbf{801}   & school area; all parks;                      \\ \cline{2-4} 
                                                                            & \textbf{Distance}         &  \textbf{650}   & Nearby                                            \\ \cline{2-4} 
                                                                            & Default          &  154   & (everywhere) ; (all) locations                    \\ \hline
\multirow{3}{*}{\begin{tabular}[c]{@{}l@{}}\textbf{Aggregation}\\ \end{tabular}} & \textbf{Count, Sum}       &  \textbf{256}   & in total; x out of N locations; \%;                    \\ \cline{2-4} 
                                                                            & \textbf{Average}          &  \textbf{196}   & per m$^2$;                                        \\ \cline{2-4} 
                                                                            & \textbf{Max, Min}         &  \textbf{67}   & highest/lowest value                              \\ \hline                               
\textbf{Entity} & Subject & 1000 & air quality; Buses; \\\hline
\multirow{3}{*}{\begin{tabular}[c]{@{}l@{}}\textbf{Comparison}\\ \end{tabular}}     & Value comparison &  836   & More than, less than                              \\ \cline{2-4} 
                                                                            & Boolean          &  388   & Street is blocked; should                                \\ \cline{2-4} 
                                                                            & Not              &   456  & It is unlawful/prohibited...                       \\ \hline
                                         
\multirow{2}{*}{\begin{tabular}[c]{@{}l@{}}\textbf{Condition}\\ \end{tabular}}  
                                                                            & Until            &    24 & keep... until the street is not blocked.          \\ \cline{2-4} 
                                                                            & If/Except           &  44   & If rainy, the speed limit...  \\ \hline
\end{tabular}
\vspace{-2em}
\end{table}




\textbf{Temporal}: Most of the requirements include a variety of temporal constraints, e.g. a static deadline, a dynamic deadline, or time intervals. In many cases (65.7\%), the temporal information is not explicitly written in the requirement, which usually means it should be ``always'' satisfied. 
In addition, city requirements are highly real-time driven. In over 80\% requirements, cities are required to detect and resolve requirement violations at runtime. It indicates a high demand for runtime monitoring. 


\textbf{Spatial}: A requirement usually specifies its spatial range explicitly using the Points of Interest (PoIs) (80.1\%), such as  ``park'', ``xx school'', along with a distance range (65\%). One requirement usually points to a set of places (e.g. all the schools). Therefore, it is very important for a formal language to be able to specify the spatial elements across many locations within the formula, rather than one formula for each location. 

We also found that the city requirements specify a very large spatial scale. Different from the requirements of many other cyber physical systems, requirements from smart cities are highly spatial-specific and usually involve a very large number of locations/sensors. 
For example, the first requirement in \tabref{tab:reqexamples} specifies a vehicle idling time ``adjacent to any school, pre-K to 12th grade in the City of New York''. There are about 2000 pre-K to 12th schools, even counting 20 street segments nearby each school, there are 40,000 data streams to be monitored synchronously. An efficient monitoring is highly demanded.

\textbf{Aggregation}: In 51.9\% cases, requirements are specified on the aggregated signal over an area, such as, ``the total amount'', ``average...per 100 square feet'', ``up to four vending vehicles in any given city block'', ``at least 20\% of travelers from all entrances should ...'', etc. Therefore, aggregation is a key feature for the specification language.  




\section{Formalizing Temporal-Spatial Requirements}
\label{sect:spec}



SaSTL extends STL with 
two spatial operators: a \emph{spatial aggregation} operator
and a \emph{neighborhood counting} operator.
Spatial aggregation enables combining (according to a chosen 
operation) measurements of the same type (e.g., environmental 
temperature), but taken from different locations.
The use of this operator can be suitable in requirements 
where it is necessary to evaluate the average, best or worst 
value of a signal measurement in an area close to the desired location.
The neighborhood counting operator allows measuring the number/percentage of neighbors of a location that satisfy a certain requirement. 
In this section, we formally define the syntax, and semantics.


\subsection{SaSTL Syntax}
We define a multi-dimensional \emph{spatial-temporal signal} as $\omega: \mathbb{T} \times L \to \{{\mathbb{R}\cup\{\bot\}\}} ^n$,
where $\mathbb{T}=\mathbb{R}_{\ge 0}$, represents the continuous time and $L$ is the set of locations. We define $X= \{x_1, \cdots, x_n \}$ as the set of 
variables for each location. 
Each variable can assume a real value $ v \in \mathbb{R}$ or
is undefined for a particular 
location ($x_i = \bot$).

We denote by $\pi_{x_{i}}(\omega)$ as the projection of $\omega$ on its component variable $x_i \in X$.  We define $P = \{p_1, \cdots, p_m \}$ a set of propositions (e.g. $\{\mathsf{School, Street, Hospital}, \cdots\}$ ) and 
$\mathcal{L}$ a labeling function $\mathcal{L}: L \rightarrow 2^{P}$ that assigns for each location 
the set of the propositions that 
are true in that location. 

A \emph{weighted undirected graph} is a tuple $G=(L, E, \eta)$ where 
$L$ is a finite non-empty set of nodes representing locations,
$E \subseteq L \times L$ is the set of edges connecting nodes,
and $\eta: E \to \mathbb{R}_{\ge 0}$ is a cost function over edges.
We define the \emph{weighted distance} between two locations $l, l' \in L$ as
$$ d(l,l'):= \min\{\sum_{e\in \sigma}\eta(e) \ |\ \sigma \mbox{ is a path between } l \mbox{ and } l'\}. $$

Then we define the spatial domain $\ra$ as, 

\begin{equation*}
\begin{array}{cl}
     \ra :=&  ([d_1, d_2],\psi) \\
     \psi :=& \top\;|\;p\;|\;\neg\;\psi\;|\;\psi\;\vee\;\psi 
\end{array}
\end{equation*}
 
where $[d_1, d_2]$ defines a spatial interval with $d_1 < d_2$ and $d_1,d_2 \in \mathbb{R}$, and $\psi$ specifies the property 
over the set of propositions
that must hold in each location.
In particular, 
$\ra = ( [0,+\infty), \top)$ indicates the whole spatial domain. 
We denote $L_{([d_1,d_2], \psi)}^l:=\{l' \in L \ | \ 0 \le d_1 \le d(l, l') \le d_2 \mbox{ and } 
\mathcal{L}(l') \sat \psi \}$ 
as the set of locations at a distance between $d_1$ and $d_2$ from $l$ for which $\mathcal{L}(l')$ satisfies $\psi$.    
We denote the set of non-null values for signal variable $x$ at time point $t$ location $l$ over locations in $\nb^l$ by 
\begin{equation*}
\scriptsize
 \nbx:=\{\pi_x(\omega)[t, l'] \ | \ l' \in \nb^l \mbox{ and } \pi_x(\omega)[t, l'] \neq \bot\}.   
\end{equation*}
We define a set of operations $\op(\nbx)$ for $\op \in \{\max, \min, \mathrm{sum}, \avg\}$  when $\nbx \neq \emptyset$ 
that computes the maximum, minimum, summation and average of values in the set $\nbx$, respectively. 


To be noted, Graph $G$ and its weights between nodes are constructed flexibly based on the property of the system. For example, we can build a graph with fully connected sensor nodes and their Euclidean distance as the weights when monitoring the air quality in a city; or we can also build a graph that only connects the street nodes when the two streets are contiguous and apply Manhattan distance. 
It does not affect the syntax and semantics of SaSTL. 


The syntax of SaSTL is given by 
\begin{equation*} 
\begin{array}{cl}
\varphi  :=& x \sim c\;| \
 \neg \varphi \;| \
\varphi_1 \land \varphi_2 \;|\
\varphi_1 \until \varphi_2 \;|\
\ag^{\op} x \sim c \;|\
\ct^{\op} \varphi \sim c \\
\end{array}
\end{equation*}

where $x \in X$, $\sim \in \{<, \le\}$, $c \in \mathbb{R}$ is a constant, 
$I \subseteq \mathbb{R}_{> 0}$ is a real positive dense time interval, 
$\until$ is the \emph{bounded until} temporal operators from STL. 
The \emph{always} (denoted $\always$) and \emph{eventually} (denoted $\eventually$) temporal operators can be derived the same way as in STL, where $\eventually \varphi \equiv \mathsf{true} \ \until \varphi$, and $\always \varphi \equiv \neg \eventually \neg \varphi$.


We define a set of spatial \emph{aggregation} operators $\ag^{\op} x \sim c$ for 
$\op \in \{\max, \min, \mathrm{sum}, \avg\}$
that evaluate the aggregated product of traces $\op(\nbx)$ over a set of locations $l \in \nb^l$.
We also define a set of new spatial \emph{counting} operators $\ct^{\op} \varphi \sim c$ for 
$\op \in \{\max, \min, \mathrm{sum}, \avg\}$ that counts the satisfaction of traces over a set of locations. 
More precisely, we define 
$\ct^{\op} \varphi = \op(\{g((\omega, t, l')  \sat \varphi) \ | \ l' \in \nb^l\})$, where $g((\omega, t, l)  \sat \varphi)) = 1$ if $(\omega, t, l)  \sat \varphi$, otherwise $g((\omega, t, l)  \sat \varphi)) = 0$. 
From the new \textit{counting} operators, we also derive the \emph{everywhere} operator as $\ew \varphi \equiv \ct^{\mathrm{min}} \varphi > 0$, and \emph{somewhere} operator as $\sw \varphi \equiv \ct^{\mathrm{max}} \varphi > 0$.
In addition, $\ct^{\mathrm{sum}} \varphi$ specifies the total number of locations that satisfy $\varphi$ and $\ct^{\mathrm{avg}} \varphi$ specifies the percentage of locations satisfying $\varphi$. 

We now illustrate how to use SaSTL to specify various city requirements, especially for the spatial aggregation and spatial counting, and how important these operators are for the smart city requirements using examples below.

\begin{example} [Spatial Aggregation]
Assume we have a requirement, ``The average noise level in the school area (within 1 km) in New York City should always be less than 50 dB and the worst should be less than 80 dB in the next 3 hours'' is formalized as,
   $\everywhere_{([0, +\infty), \mathsf{School})}\always_{[0,3]}( (\agr_{([0,1], \top)}^{\avg} x_\mathsf{Noise} < 50) \land (\agr_{([0,1], \top)}^{\mathsf{max}} x_\mathsf{Noise} < 80))$.
\label{ex:sytax}
\vspace{-0.5em}
\end{example}

$([0, +\infty), \mathsf{School})$ selects all the locations labeled as ``school'' within the whole New York city (${[0, +\infty)}$) (predefined by users). 
$\always_{[0, 3]}$ indicates this requirement is valid for the next three hours. $(\agr_{([0,1], \top)}^{\avg} x_\mathsf{Noise} < 50) \land (\agr_{([0,1], \top)}^{\mathsf{max}} x_\mathsf{Noise} < 80)$ calculates the average and maximal values in 1 km for each ``school'', and compares them with the requirements, i.e. 50 dB and 80 dB. 

Without the spatial aggregation operators, STL and its extended languages cannot specify this requirement. First, they are not able to first dynamically find all the locations labeled as ``school''. 
To monitor the same spatial range, users have manually get all traces from schools, and then repeatedly apply this requirement to each located sensor within 1 km of a school and do the same for all schools.  
More importantly, STL and its extended languages could not specify ``average'' or  ``worst'' noise level. Instead, it only monitors each single value, which is prone to noises and outliers and thereby causes inaccurate results.

\begin{example}[Spatial Counting]
A requirement that ``At least 90\% of the streets, the particulate matter (PMx) emission should not exceed \textit{Moderate} in 2 hours'' is formalized as
  $\mathcal{C}_{([0,+\infty), \mathsf{Street})}^\avg(\always_{[0,2]} ( x_\mathsf{PMx} < \mathsf{Moderate})) > 0.9$.
\label{ex2:syntax}
\end{example}
${\mathcal{C}_{([0,+\infty),\mathsf{Street})}^\avg} \varphi > 0.9$ represents the percentage of satisfaction is larger than 90\%. 
Specifying the percentage of satisfaction is very common and important among city requirements. 





\subsection{SaSTL Semantics}
We define the SaSTL semantics as the \emph{satisfiability relation}
$(\omega, t, l) \sat \varphi$, indicating that the spatio-temporal signal $\omega$ satisfies a formula $\varphi$ at the time point $t$ in location $l$ when $\pi_v(\omega)[t, l]\neq \bot$ and $\nbx \neq \emptyset$.  We define that $(\omega, t, l)  \sat \varphi$ if $\pi_v(\omega)[t, l] = \bot$.

\vspace{-1em}
\begin{alignat*}{2}
    (\omega, t, l) & \sat x \sim c 
        && \eqdef \pi_x(\omega)[t, l] \sim c  \\
    (\omega, t, l) & \sat \neg \varphi
        && \eqdef (\omega, t, l) \not \sat \varphi \\
    (\omega, t, l) & \sat \varphi_1 \land \varphi_2
        && \eqdef (\omega, t, l) \sat \varphi_1 \mbox{ and } (\omega, t, l) \sat \varphi_2 \\
    (\omega, t, l) & \sat \varphi_1 \until \varphi_2 
        && \eqdef \exists t' \in (t+I) \cap \mathbb{T}: (\omega, t', l) \sat \varphi_2 \\
        & && \hspace{2em}  \mbox{ and } \forall t'' \in (t, t'), (\omega, t'', l) \sat \varphi_1 \\
    (\omega, t, l) & \sat  \ag^{\op} x \sim c   
        && \eqdef \op (\nbx  ) \sim c  \\
    (\omega, t, l) & \sat  \ct^{\op} \varphi \sim c   
        && \eqdef     \op(\{g((\omega, t, l')  \sat \varphi) \ | \ l' \in \nb^l\}) \sim c \\
\end{alignat*}

\section{Efficient Monitoring for SaSTL}
\label{sect:alg}

In this section, we first present monitoring algorithms for SaSTL, then describe two optimization methods to speed up the monitoring performance.


\subsection{Monitoring Algorithms for SaSTL}

The \textit{inputs} of the monitor are the SaSTL requirements 
$\varphi$ (including the time $t$ and location $l$), a weighted undirected graph $G$ and the temporal-spatial data $\omega$. The \textit{output} of the monitoring algorithm for each requirement
is a Boolean value indicating whether the requirement is 
satisfied or not. \algref{alg:sastl} outlines the monitoring algorithm.
To start with, the monitoring algorithm parses $\varphi$ to sub formulas and calculates the satisfaction for each operation recursively. 
We derived operators $\always$ and $\eventually$ from $\until$, and operators $\everywhere$ and $\somewhere$ from $\ct^{\mathsf{op}}\sim c$, so we only show the algorithms for $\until$ and $\ct^{\mathsf{op}}\sim c$.  
\begin{algorithm}
\tablefontsize
  \SetKwFunction{Monitor}{Monitor}
  \SetKwProg{Fn}{Function}{:}{}
  \Fn{\Monitor{$\varphi,\omega, t, l, G$}}{
      \SetKwInOut{Input}{Input}
      \SetKwInOut{Output}{Output}
      \SetKwFor{Case}{Case}{}{}
      \Input{SaSTL Requirement $\varphi$, Signal $\omega$, Time $t$, Location $l$, weighted undirected graph $G$}
      
      \Output{Satisfaction Value (Boolean value)}
      
      \Begin{
            \Switch{$\varphi$} {

                   \Case{$x\sim c$}{
                        \Return $\pi_x(\omega)[t, l] \sim c$;
                   }

                   \Case{$\neg \varphi$}{
                        \Return $\neg$ Monitor($\varphi,\omega, t, l, G$);
                   }
 
                   \Case{$\varphi_1 \land \varphi_2$  \Comment*[r]{See \algref{alg:and}}}{
                        \Return Monitor($\varphi_1,\omega, t, l, G$) 
                        $\land$ Monitor($\varphi_2,\omega, t, l, G$) 
                    }
                    
                    \Case{$ \varphi_1 U_I \varphi_2$ 
                    }{ 
                         \Return SatisfyUntil($\varphi_1,\varphi_2, I, \omega, t, l, G$);
                    }
 

                    \Case{$\ag^{\op} x \sim c$   \Comment*[r]{See \algref{alg:Aggregation}} }
                    {
                          \Return Aggregate$(x, c, op, \ra, t, l, G)$;
                    }
                    
                    \Case{$\ct^{\op} \varphi \sim c$   \Comment*[r]{See \algref{alg:Counting} for the standard version, and \algref{alg:CountingPara} for an improved parallel version.}}
                    { 
                          \Return CountingNeighbours$(\varphi, c, op, \ra, t, l, G) $; 
                    }
           }
       }
    }
\caption{SaSTL monitoring algorithm Monitor({$\varphi,\omega, t, l, G$})}
\label{alg:sastl}
\end{algorithm}

We present the satisfaction algorithms of the operators $\ag^{\mathsf{op}}$ and $\ct^{\mathsf{op}}$  in \algref{alg:Aggregation} and \algref{alg:Counting}, respectively. As we can tell from the algorithms, essentially, $\ag^{\mathsf{op}}$ calculates the aggregated values on the signal over a spatial domain, while $\ct^{\mathsf{op}}$ calculates the aggregated Boolean results over spatial domain. 

The time complexity of monitoring the logical and temporal operators of SaSTL is the same as STL~\cite{donze2013efficient} as follows.
\begin{itemize}
    \item The time complexity to monitor classical logical operators or basic propositions such as $\neg x$, $\land$ and $x\sim c$ is $O(1)$.
    \item The time complexity to monitor  temporal operators such as $\always_{I}$, $\eventually_{I}$, $\until$ is $O(T)$, where $T$ is the total number of samples within time interval $I$.
\end{itemize}

\begin{algorithm}
\tablefontsize
 \SetKwFunction{CS}{Aggregate}

  \SetKwProg{Fn}{Function}{:}{}
  \Fn{\CS{$x, c, op, \ra, \omega, t, l, G$}}{
  \Begin{
        \textbf{Real} v := 0; n := 0;
        
        \lIf{$op$ == "min"}{
             $v := \infty $
         }
        \lIf{$op$ == "max"}{
             $v := - \infty $
         }

     \For {$l' \in L^l_{\ra}$}{
          
           \If{$\mathsf{op} \in \{$min, max, sum$\}$}{
                       $v$ := $\mathsf{op}(v,\pi_x(\omega)[t, l'])$;
                  }
                  \If{$\mathsf{op} == $"avg"}{
                      $v$ := $\mathsf{sum}(v,\pi_x(\omega)[t, l'])$;
                  }
                  $n := n+1$
     }
     \lIf{$\mathsf{op}$ == "avg" $\land n \neq 0$}{
                $v :=v / n$
            }
    \eIf{$n==0$}{\Return $\mathsf{True}$}{\Return $v \sim c$;}
     
  }

  }
  
\caption{Aggregation of $(x, op, \ra, \omega, t, l, G)$}
\label{alg:Aggregation}
\end{algorithm}

\begin{algorithm}
\tablefontsize
 \SetKwFunction{CS}{CountingNeighbours}

  \SetKwProg{Fn}{Function}{:}{}
  \Fn{\CS{$\varphi, c, op, \ra, \omega, t, l, G$}}{
      \Begin{
              \textbf{Real} $v := 0$; $n := 0$
              
              \lIf{$op$ == "min"}{
                    $v := \infty$
              }
              
              \lIf{$op$ == "max"}
              {
                  $v := - \infty$
              }
             
              \For {$l' \in L^l_{\ra}$}{
                  \If{Monitor$(\varphi,\omega, t, l, G)$ $\land$ $\mathsf{op} \in \{$min, max, sum$\}$}{
                      $v$ := $\mathsf{op}(v,1)$;
                  }
                  \If{Monitor$(\varphi,\omega, t, l, G)$ $\land$ $\mathsf{op} == $"avg"}{
                      $v$ := $\mathsf{sum}(v,1)$;
                  }
                  $n := n+1$
             }

            \lIf{$\mathsf{op}$ == "avg" $\land n \neq 0$}{
                $v :=v / n$
            }
        \eIf{$n==0$}{\Return $\mathsf{True}$}{\Return $v \sim c$;}

      }
      
  }
\caption{Counting of $(x, op, \ra, \omega, t, l, G)$}
\label{alg:Counting}
\end{algorithm}

In this paper, we present the time complexity analysis for the spatial operators (Lemma \ref{lemma:spatial}) and the new SaSTL monitoring algorithm (Theorem \ref{th:timeAlg}). 
The total number of locations is denoted by $n$.  We assume that the positions of 
the locations cannot change in time (a fixed grid). We can precompute all the distances between locations and store them in 
an array of range trees~\cite{Lueker78} (one range tree for each 
location).
We further denote the monitored formula as $\phi$, which can be represented by a syntax tree, and let $|\phi|$ denote the total number of nodes in the syntax tree (number of operators). 
 
\begin{lemma}[Complexity of spatial operators]

The time complexity to monitor at each location $l$ at time $t$ the satisfaction of a spatial operator such as  $\everywhere_{\ra}$, $\somewhere_{\ra}$, $\ag^\mathsf{op}$, and $\ct^\mathsf{op}$ is $O(log(n) + |L|)$ 
where L is the set of locations at distance within 
the range $\ra$ from $l$.
\label{lemma:spatial}
\end{lemma}
\begin{proof}
According to~\cite{Lueker78}, the time complexity to 
retrieve a set of nodes $L$ with a distance 
to a desired location in a range $\ra$  from a location $l$ is $O(log(n) + |L|)$.
The aggregation and counting operations of \algref{alg:Aggregation} and \algref{alg:Counting} can be performed
while the locations are retrieved.
\end{proof}


\begin{mythm}
The time complexity of the SaSTL monitoring algorithm  is upper-bounded by $O(|\phi|~T_{max}~(log(n) + |L|_{max}))$
where $T_{max}$ is the largest number of samples of  the 
intervals considered in the temporal operators of $\phi$ 
and $|L|_{max}$ is the maximum number of locations defined 
by the spatial temporal operators of $\phi$. 
\label{th:timeAlg}
\end{mythm}

\begin{proof}
Following Lemma \ref{lemma:spatial}, by considering 
$T_{max}$ the worst possible number of samples 
that we need to consider for all possible intervals of temporal operators present
in the formula, and $|L|_{max}$ for the worst possible number of locations that we need to consider for all possible intervals of spatial operators present
in the formula. 
When there are two or more operators nested, the time complexity for one operation is bounded by $O(T_{max}~(log(n) + |L|_{max}))$. 
As there are $|\phi|$ nodes in the syntax tree of $\phi$, the time complexity of the SaSTL monitoring algorithm is bounded by the summation over all $|\phi|$ nodes, which is $O(|\phi|~T_{max}~(log(n) + |L|_{max}))$.



\end{proof}

\subsection{Performance Improvement of SaSTL Parsing}

To monitor a requirement, the first step is parsing the requirement to a set 
of sub formulas with their corresponding spatial-temporal ranges. 
Then, we calculate the results for the sub formulas. The traditional parsing process of STL builds and calculates the syntax tree on the sequential order of the formula. It does not consider the complexity of each sub-formula. 
However, in many cases, especially with the PoIs specified in smart cities, checking the simpler propositional variable to quantify the spatial domain first can significantly reduce the number of temporal signals to check in a complicated formula. 
For example, the city abstracted graph in \figref{fig:eff}, the large nodes represent the locations of PoIs, among which the red ones represent the schools, and blue ones represent other PoIs. The small black nodes represent the locations of data sources (e.g. sensors). Assuming a requirement $\everywhere_{([0, +\infty], \mathsf{School})} \always_{[a,b]} (\agr^\mathsf{op}_{([0,d],\top])} \varphi \sim c)$ requires to aggregate and check $\varphi$ only nearby schools (i.e., the red circles), but it will actually check data sources of all nearby 12 nodes if one is following the traditional parsing algorithm. 
In New York City, there are about 2000 primary schools, but hundreds of thousands of PoIs in total. A very large amount of computing time would be wasted in this way.


\begin{figure}[t]
    \centering
     \includegraphics[width=5.6cm]{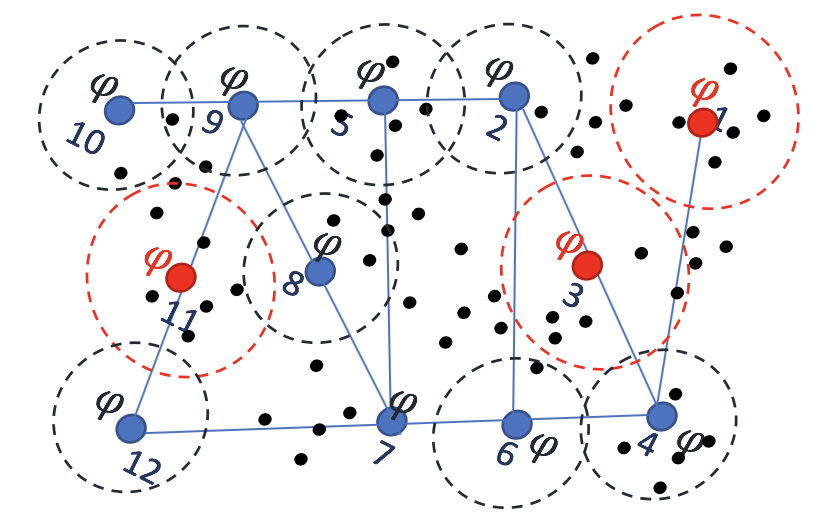} \\
    \caption{An example of city abstracted graph. A requirement is $\everywhere_{([0, +\infty], \mathsf{School})} \always_{[a,b]} (\agr^\mathsf{op}_{([0,d], \top)} \varphi \sim c)$ (The large nodes represent the locations of PoIs, among which the red ones represent the schools, and blue ones represent other PoIs. The small black nodes represent the locations of data sources.)}
     \vspace{-0.6em}
    \label{fig:eff}
\end{figure}

To deal with this problem,
we now introduce a monitoring cost function $\mathsf{cost}: \Phi \times L \times G_L \rightarrow \mathbb{R}^+ $, where $\Phi$ is the set of all the possible SaSTL formulas,
$L$ is the set of locations, $G_L$ is the set of all the possible undirected graphs with $L$ locations. The cost function for $\varphi$ is defined as:

\begin{equation*}
\footnotesize
\begin{split}
    &\mathsf{cost}(\varphi, l, G) =  \\
    &    \begin{cases}
        1 & \mathsf{if}~\varphi:=p \lor \varphi:=x\sim c \lor \varphi:=\mathsf{True} \\
         1+\mathsf{cost}(\varphi_1, l, G) & \mathsf{if}~\varphi:= \neg \varphi_1\\
         \mathsf{cost}(\varphi_1, l, G) + \mathsf{cost}(\varphi_2, l, G) & \mathsf{if}~\varphi := \varphi_1 * \varphi_2, *\in\{ \land, \until\}\\
         |L^l_{\ra}| &\mathsf{if}~ \varphi := \ag^\mathsf{op} x\sim c \\
         |L^l_{\ra}| \mathsf{cost}(\varphi_1, l, G)&\mathsf{if}~ \varphi := \ct^\mathsf{op} \varphi_1 \sim c\\
        \end{cases}
\end{split}
\end{equation*}

Using the above function, the cost of each operation is calculated before ``switch $\varphi$'' (refer \algref{alg:sastl}). 
The cost function measures how complex it is to monitor 
a particular SaSTL formula.  This can be used when the algorithm 
evaluates the $\land$ operator and it establishes the order 
in which the sub-formulas should be evaluated. The simpler 
sub-formula is the first to be monitored, while the more 
complex one is monitored only when the other sub-formula 
is satisfied. We update $\mathsf{monitor}(\varphi_1 \land \varphi_2,\omega)$ in \algref{alg:and}.

\begin{algorithm}
\tablefontsize
  \Case{$\varphi_1 \land \varphi_2$ }{
                        \Return Monitor($\varphi_1,\omega, t, l, G$) 
                        $\land$ Monitor($\varphi_2,\omega, t, l, G$); \\
                        \If{$\mathsf{cost}(\varphi_1, l, G) \leq \mathsf{cost}(\varphi_2, l, G)$}{      
                            \If {$\neg$ Monitor($\varphi_1,\omega, t, l, G$)}{
                                 \Return Monitor($\varphi_2,\omega, t, l, G$);
                            }
                            \Return True;
                        }
                        \If {$\neg$ Monitor($\varphi_2,\omega, t, l, G$)}{
                            \Return Monitor($\varphi_1,\omega, t, l, G$);
                        }
                        \Return True;
                    }
\caption{Satisfaction of $(\varphi_1 \land \varphi_2, \omega)$}
\label{alg:and}
\end{algorithm}


With this cost function, the time complexity of the monitoring algorithm is reduced to $O(|\phi|T_{max}(log(n)+|L'|_{max}))$, where $|L'|$ is the maximal number of locations that an operation is executed with the improved parsing method. The improvement is significant for monitoring smart cities, because for most of the city requirements, $|L'|_{max} < 100 \times |L|_{max}$.

\subsection{Parallelization}
In the traditional STL monitor algorithm, the signals are checked sequentially. For example, to see if the data streams from all locations satisfy $\everywhere_{\ra}\always_{[a,b]}\varphi$ in \figref{fig:eff}, usually, it would first check the signal from location 1 with $\always_{[a,b]}\varphi$, then location 2, and so on. At last, it calculates the result from all locations with $\everywhere_{\ra}$. In this example, checking all locations sequentially is the most time-consuming part, and it could reach over 100 locations in the field. 


To reduce the computing time, we parallelize the monitoring algorithm in the spatial domain. To briefly explain the idea: instead of calculating a sub-formula ($\always_{[a,b]}\varphi$) at all locations sequentially, we distribute the tasks of monitoring independent locations to different threads and check them in parallel. 
(\algref{alg:CountingPara} presents the parallel version of the spatial counting operator $\ct$.) To start with, all satisfied locations $l' \in L^l_{\ra}$ are added to a task pool (a queue). In the mapping process, each thread retrieves monitoring tasks (i.e., for $l_i, \always_{[a,b]}\varphi$) from the queue and executes them in parallel. All threads only execute one task at one time and is assigned a new one from the pool when it finishes the last one, until all tasks are executed. Each task obtains the satisfaction of $\mathsf{Monitor}(\varphi, \omega, t, l, G)$ function, and calculates the local result $v_i$ of operation $\mathsf{op}()$. The reduce step sums all the parallel results and calculates a final result  of  $\mathsf{op}()$. 


\begin{algorithm}
\tablefontsize
 \SetKwFunction{CS}{CountingNeighbours}
 \SetKwFunction{worker}{worker}
  \SetKwProg{Fn}{Function}{:}{}
  \Fn{\CS{$\varphi,op, \ra, \omega, t, l, G$}}{
      \Begin{
 
              paratasks = Queue();

            \For{$l' \in L^l_{\ra}$}{paratasks.add($l$)\;}
                
                results = Queue()\;
                
                \For{i in $1..\text{NumThreads}$}{
                Thread$_i$ $\leftarrow$ worker($\varphi,\omega, t, G$)\;
                }
                Wait()\;
                \Return op(results)\;
      }
}
  \Fn{\worker($\varphi,\omega, t, G$)}{
  \Begin{
               \textbf{Real} $v := 0$;
              
              \lIf{$op$ == "min"}
              {
                    $v := \infty $;
              }
              
              \lIf{$op$ == "max"}
              {
                  $v := - \infty $;
              }
\While{Num(tasks)>0}{
$l$ = paratasks.pop()\;
moni = Monitor($\varphi, \omega, t, l, G$)\;
v = op(v, moni)\;
}
  results.add(v)
  }
}
            
\caption{Parallelization of Counting of $(x, op, \ra, \omega, t, l, G)$}
\label{alg:CountingPara}
\end{algorithm}

\begin{lemma}
The time complexity of the parallelized algorithm Monitor($\phi$, $\omega$) is upper bounded by $O({|\phi|}T_{max}(log(n)+\frac{|L|_{max}}{P}))$ when distributed to $P$ threads.
\label{lemma:parallel}
\end{lemma}



In general, the parallel monitor on the spatial domain reduces the computational time significantly. It is very helpful to support runtime monitoring and decision making, especially for a large number of requirements to be monitored in a short time. 
In practice, the computing time also depends on the complexity of temporal and spatial domains as well as the amount of data to be monitored. 
A comprehensive experimental analysis of the time complexity is presented in \sectref{sect:eva}.


\section{Evaluation}
\label{sect:eva}

We evaluate the SaSTL monitor by applying it to two city application scenarios, which are (i) runtime monitoring of city requirements in Chicago, and (ii) runtime conflict detection and resolution among smart services in a simulated New York City. Then, (iii) we compared the coverage expressiveness of SaSTL against 1000 real city requirements with STL, SSTL, and STREL. 
We provide the information of two application scenarios in \tabref{tab:three_city}.
\figref{fig:maps} presents the partial maps of two cities, where the locations of PoIs and sensors are marked.

\begin{figure}[t]
    \centering
     \includegraphics[width = 0.2\textwidth, trim=0 2cm 0 0, clip]{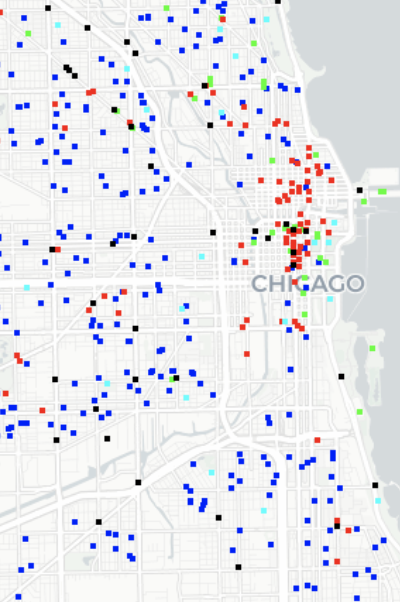} \hspace{0.3cm}
\includegraphics[width = 0.2\textwidth, trim=0 2cm 0 0, clip]{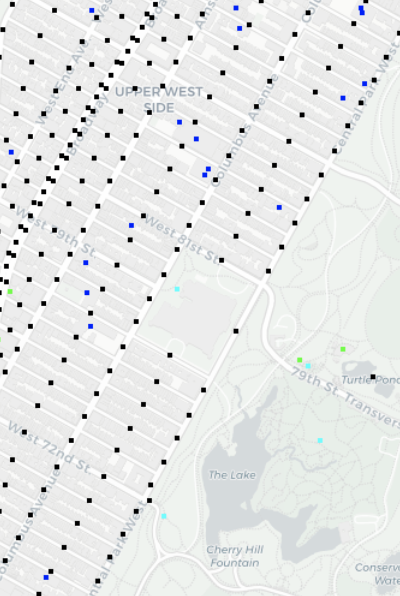}\\
\scriptsize{
(1) Chicago    \hspace{3cm}      (2) New York }
    \caption{Partial Maps of Chicago and New York with PoIs and sensors annotated. (The black nodes represent the locations of sensors, red nodes represent the locations of hospitals, dark blue nodes represent schools, light blue nodes represent parks and green nodes represent theaters.)}
    \label{fig:maps}
    \vspace{-1.5em}
\end{figure}

\begin{table}[t]
\caption{Information of Two Application Scenarios}
\label{tab:three_city}
\tablefontsize
\begin{tabular}{|L{1.5cm} |L{2.5cm}| L{3.5cm}|}
\hline
           & \textbf{Chicago}   &\textbf{ New York }                                                                                                                       \\\hline
\textbf{Area (km${^2}$) }    & 606 & 60                                                                                                                                                  \\ \hline
\textbf{Data Type}        & Real-time States     & Real-time Predicted States                                                                                                                             \\ \hline
\textbf{Data Sources} &  Real Sensors & Simulated Sensors and Actuators   \\\hline
\textbf{Number of Sensor Nodes}     & 118    & 10,000                                                                                                                                               \\  \hline
\textbf{Time Period}    & 2017.01-2019.05     & -                                                                                                                                        \\ \hline
\textbf{Sampling Rate}   & 1 min    & 10 seconds                                                                                                                                      \\ \hline
\textbf{Domain}     & Environment, Public Safety      & Environment, Transportation, Events, Emergencies, Public Safety                       \\\hline
\textbf{Variables}      & CO, NO, O$_3$, Visible light, Crime Rate & CO, NO, O$_3$, PM$_x$, Noise, Traffic, Pedestrian, Signal Lights, Emergency Vehicles, Accidents    \\ \hline
\end{tabular}
\vspace{-1.5em}
\end{table}



\subsection{Runtime Monitoring of Real-Time Requirements in Chicago}

\subsubsection{Introduction}
We apply SaSTL to monitor the real-time requirements in Chicago. The framework is the same as shown in \figref{fig:overview}, where we first formalize the city requirements to SaSTL formulas and then monitor the city states with the formalized requirements. Chicago is collecting and publishing the city data \cite{arrayofthings, chicagocrime} every day since January, 2017 without a state monitor.  
In this evaluation, we emulate the real data as it arrives in real time, i.e. assuming the city was operating with the monitor. 
Then we specify 80 safety and performance requirements that are generated from the real requirements, and apply the SaSTL to monitor the data every 3 hours continuously to identify the requirement violations.

\subsubsection{Chicago Performance}

Valuable information is identified from the monitor results of different periods during a day.
We randomly select 30 days of weekdays and 30 days of weekends. We divide the daytime of a day into 4 time periods and 3 hours per time period. 
We calculate the percentage of satisfaction (i.e., number of satisfied requirement days divides 30 days) for each time period, respectively. The results of two example requirements CR1 and CR2 are shown in \figref{fig:chicago_case}. 
CR1 specifies ``The average air quality  within 5km of all schools should always be above \textit{Moderate} in the next 3 hours.'' and is formalized as $\everywhere_{([0,+\infty),  {\mathsf{School}})}\always_{[0,3]}(\agr_{([0,5],\top)}^{\avg} x_{\mathsf{air}} > \mathsf{Moderate})$. CR2 specifies ``For the blocks with a high crime rate, the average light level within 3 km should always be \textit{High}'' and is formalized as $\everywhere_{([0,+\infty),\top)}\always_{[0,3]}  ( x_{\mathsf{Crime}} = \mathsf{High} \rightarrow \agr_{([0,3],\top)}^{\avg} x_{\mathsf{Light}} >= \mathsf{High})$. 

\vspace{-0.2em}
\begin{figure}[t]
    \centering
    \includegraphics[width=.5\textwidth]{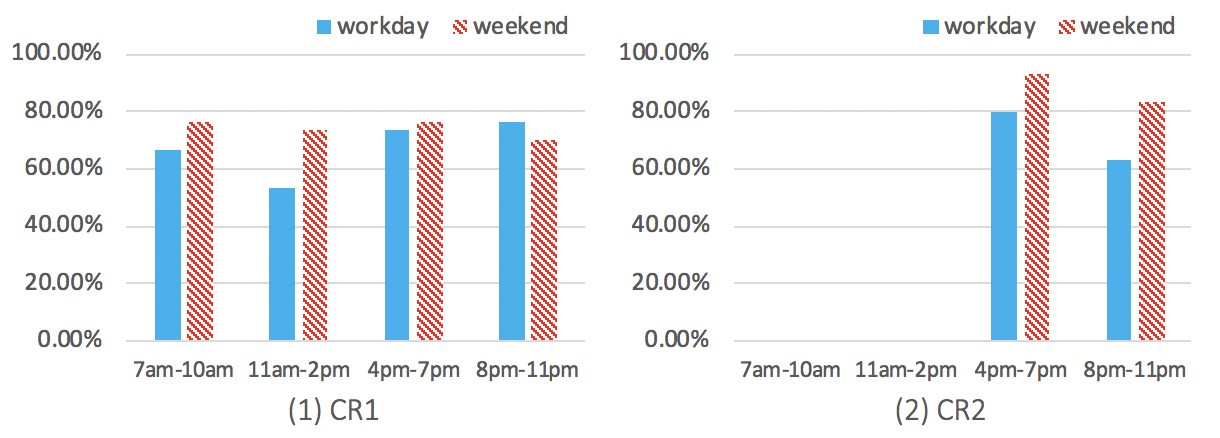}
    \caption{Requirement Satisfaction Rate during Different Time Periods in Chicago}
    \label{fig:chicago_case}
    
\end{figure}

The SaSTL monitor results can be potentially used by different stakeholders. 

First, \textit{with proper requirements defined, the city decision makers are able to identify the real problems and take actions to resolve or even avoid the violations in time}. For example, from the two example requirements in \figref{fig:chicago_case}, we could see over 20\% of the time the requirements are missed everyday. 
Based on the monitoring results of requirement CR1, decision makers can take actions to redirect the traffic near schools and parks to improve the air quality. 
Another example of requirement CR2, the satisfaction is much higher (up to 33\% higher in CR2, 8pm - 11pm) over weekends than workdays.  There are more people and vehicles on the street on weekends, which as a result also increases the lighted areas. However, as shown in the figure, the city lighting in the areas with high crime rate is only 60\%. An outcome of this result for city managers is that they should pay attention to the illumination of workdays or the areas without enough light to enhance public safety.  

Second, \textit{it gives the citizens the ability to learn the city conditions and map that to their own requirements}. They can make decisions on their daily living, such as the good time to visit a park. For example, requirement CR1, 11am - 2pm has the lowest satisfaction rate of the day. 
The instantaneous air quality seems to be fine during rush hour, but it has an accumulative result that affects citizens' (especially students and elderly people) health. 
A potential suggestion for citizens who visit or exercise in the park is to avoid 11am - 2pm.

\subsubsection{Algorithm Performance}

\vspace{-0.3em}
\begin{figure}[t]
    \centering
    \includegraphics[width=8cm]{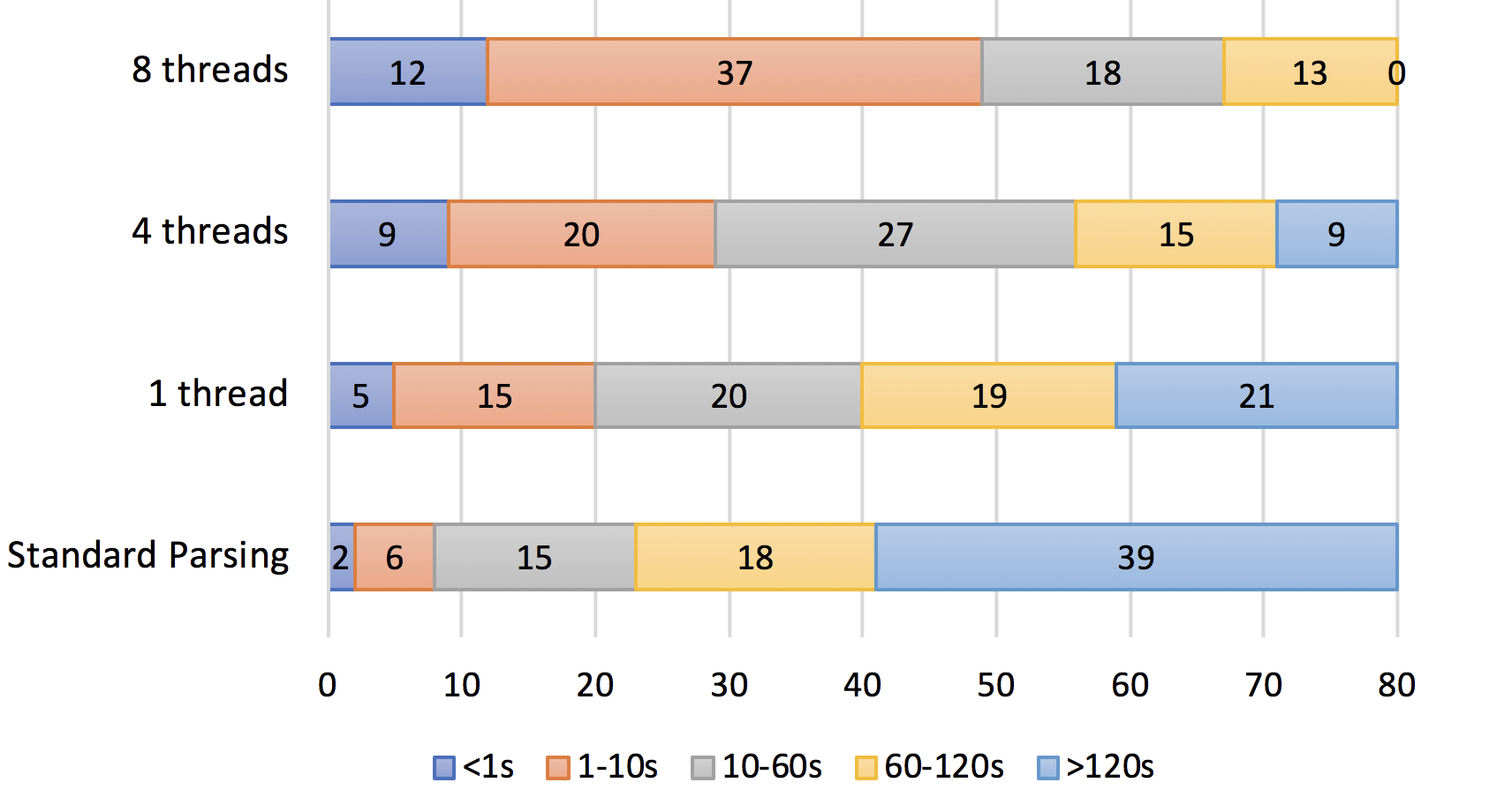}\\
    \caption{Number of Requirements Checked on Different Computing Time (x: the number of requirements, y: the number of threads, bars: different periods of computing time)}
    \label{fig:chicago80}
    \vspace{-0.3em}
\end{figure}

\begin{table*}[htbp]
\caption{Safety and Performance Requirements for New York City}
\label{table:req_ny}
\tablefontsize
\centering
\begin{tabular}{|L{0.5cm}|m{9.5cm}|L{6.5cm}|}
\hline
 & \textbf{Requirement} & \textbf{SaSTL}                                                                                                           \\\hline
\textbf{NYR1} & The average noise level in the school area (within 1km) should always be less than 50dB in the next 30min.  &  
$\boxbox_{ ([0, +\infty), \mathsf{School})}\always_{[0,30]} (\agr_{([0,1], \top)}^{\avg} x_\mathsf{Noise} < 50) $ \\\hline
\textbf{NYR2} & If an accident happens, at least one of the nearby hospitals  (within 5km), its traffic condition within 2km should not reach the level of congestion in the next 60 min. &
$\everywhere_{([0,+\infty], \top)}(\mathsf{Accident}\rightarrow 
 \mathcal{C}_{([0,5],\mathsf{Hospital})}(\always_{[0,60]}(\agr_{([0,2], \top)}^{\avg} x<\mathsf{Congestion}))>0)$
\\\hline
\textbf{NYR3} & If there is an event, the max number of pedestrians waiting in an intersection should not be greater than 50 for more than 10 minutes.     & 
$\everywhere_{([0, +\infty),\top)}(\mathsf{Event} \rightarrow  \always_{[0,10]} (\agr_{([0,1],\top)}^{\max} x_\mathsf{ped} < 50))$
\\\hline
\textbf{NYR4} & At least 90\% of the streets, the PMx emission should not exceed \textit{Moderate} in 60 min.                  &  $ \mathcal{C}_{([0,+\infty),\top)}^\avg(\always_{[0,60]} (\agr_{([0,1],\top)}^{\max} x_\mathsf{PMx} < \mathsf{Moderate})) > 0.9$\\\hline
\textbf{NYR5} & If an accident happens, it should be solved within 60 min, and before that nearby (500 m) traffic should be above moderate on average and safe in worst case.    &    
$\boxbox_{([0, +\infty),\top)}(\mathsf{Accident}\rightarrow  (\agr_{([0,500],\top)}^{\avg} x_\mathsf{traffic} <\mathsf{Moderate} \land \agr_{([0,500],\top)}^{\max} x_\mathsf{traffic} < \mathsf{Safe}) \mathcal{U}_{[0,60]} \neg \mathsf{Accident})$
\\\hline                                    
\end{tabular}
\vspace{-0em}
\end{table*}


We count the average monitoring time taken by each requirement when monitoring for 3-hour data. Then, we divide the computing time into 5 categories, i.e., less than 1 second, 1 to 10 seconds, 10 to 60 seconds, 60 to 120 seconds, and longer than 120 seconds, and count the number of requirements under each category under the conditions of standard parsing, improved parsing with single thread, 4 threads, and 8 threads.  The results are shown in \figref{fig:chicago80}. Comparing the 1st (standard parsing) and 4th (8 threads) bar, without the improved monitoring algorithms,  for about 50\% of the requirements, each one takes more than 2 minutes to execute.  
The total time of monitoring all 80 requirements is about 2 hours, which means that the city decision maker can only take actions to resolve the violation at earliest 5 hours later. 
However, with the improved monitoring algorithms, for 49 out of 80 requirements, each one of them is executed within 60 seconds, and each one of the rest requirements is executed within 120 seconds.
The total execution time is reduced to 30 minutes, which is a reasonable time to handle as many as 80 requirements. More importantly, it illustrates the effectiveness of the parsing function and parallelization methods. Even if there are more requirements to be monitored in a real city, it is doable with our approach by increasing the number of processors.

\subsection{Runtime Conflict Detection and Resolution in Simulated New York City}

\subsubsection{Introduction}
The framework of runtime conflict detection and resolution~\cite{ma2016detection} considers a scenario where smart services send action 
requests to the city center, and where a simulator predicts 
how the requested actions change the current city states over 
a finite future horizon of time. In this way, we use SaSTL monitor to specify requirements and check the \textit{predicted spatial-temporal data} with the SaSTL formulas. 
If there exists a requirement violation within the future horizon, a conflict is detected. To resolve the conflict, it provides several possible resolution options and predicts the outcome of all these options, which are verified by the SaSTL monitor.
It finds the best option without a violation, otherwise, it provides the trade-offs between a few potential resolution 
options to the city manager through monitoring the predicted traces. Details of the resolution are not the main part of this paper and we thank the original authors for making the solution available to the authors of the present paper.




We set up a smart city simulation of New York City using the Simulation of Urban MObility (SUMO) \cite{sumo} with the real city data \cite{nycopendata}, on top of which, we implement 10 smart services (S1: Traffic Service, S2: Emergency Service, S3: Accident Service, S4: Infrastructure Service, S5: Pedestrian Service, S6: Air Pollution Control Service, S7: PM2.5/PM10 Service, S8: Parking Service, S9: Noise Control Service, and S10: Event Service). 
The states from the domains of environment, transportation, events and emergencies are generated from about 10000 simulated nodes. Please see \tabref{tab:three_city} for the variables. 

We use the STL Monitor as the \textit{baseline} to compare the capability of requirement specification and the ability to improve city performance. 
We simulate the city running for 30 days in three control sets, one without any monitor, one with the STL monitor and one with the SaSTL monitor. 
For the first set (no monitor), there is no requirement monitor implemented.
For the second one (STL monitor), we specify NYR1 to NYR5 using STL without aggregation over multiple locations, i.e., NYR1 is a set of requirements on many locations and each location has a temporal requirement. For example, NYR1 is specified as $\varphi_{l_i} = \always_{[0,t]} (x_{\mathsf{air}}>\mathsf{Moderate})$, where $l_i \in L$, $L$ is a set of sensors within the range. The setting is the same for the rest of the requirements. For the third one (SaSTL monitor), five examples of different types of real-time requirements and their formalized SaSTL formulas are given in \tabref{table:req_ny}.

\subsubsection{NY City Performance}

The results are shown in \tabref{tab:performance}. We measure the city performance from the domains of transportation, environment, emergency and public safety using the following metrics, the total number of violations detected, the average CO (mg) emission per street, the average noise (dB) level per street, the emergency vehicles waiting time per vehicle per intersection, the average number and waiting time of vehicles waiting in an intersection per street, and the average pedestrian waiting time per intersection. To be noted, the number of violations detected is the total number of safety requirements violated, rather than the number of conflicts. Many times there is more than one requirement violated in one conflict. 

\begin{table}[htbp]
\caption{Comparison of the City Performance with the STL Monitor and the SaSTL Monitor}
\label{tab:performance}
\tablefontsize
\begin{tabular}{|l|R{1.2cm}|R{1.2cm}|R{1.2cm}|}
\hline
                            & {No Monitor }   & {STL Monitor} & \textbf{SaSTL Monitor} \\ \hline
\textbf{Number of Violation}         & Unknown & 219         & \textbf{173}           \\ \hline
\textbf{Air Quality Index}                     & 67.91   & 51.22       & \textbf{40.18}         \\ \hline
\textbf{Noise (db) }                 & 73.32   & 49.27       & \textbf{41.42}         \\ \hline
\textbf{Emergency Waiting Time (s)} & 20.32   & 12.93       & \textbf{11.88}        \\ \hline
\textbf{Vehicle Waiting Number }    & 22.7    & 18.2        & \textbf{12.6}          \\ \hline
\textbf{Pedestrian Waiting Time (s)}& 190.2   & 130.9       & \textbf{61.1}          \\ \hline
\textbf{Vehicle Waiting Time (s)}  & 112.12  & 70.98       & \textbf{59.22}         \\ \hline
\end{tabular}
\vspace{-0.5em}
\end{table}

We make some observations by comparing and analyzing the monitoring results. 
First, \textit{the SaSTL monitor obtains a better city performance with fewer number of violations detected under the same scenario.}
As shown in \tabref{tab:performance}, the number of violations detected by SaSTL Monitor is 46 less than the STL monitor.
On average, the framework of conflict detection and resolution with the SaSTL monitor improves the air quality by 21.6\% and  40.8\% comparing to the one with the STL monitor and without a monitor, respectively. It improves the pedestrian waiting time by 16.6\% and 47.2\% comparing to the STL monitor and without a monitor, respectively.

Second, \textit{the SaSTL monitor reveals the real city issues, helps refine the safety requirements in real time and supports improving the design of smart services. }
We also compare the number of violations on each requirement. The results (\figref{fig:pie} (1)) help the city managers to measure city's performance with smart services for different aspects, and also help policymakers to see if the requirements are too strict to be satisfied by the city and make a more realistic requirement if necessary. For example, in our 30 days simulation, apparently, NYR4 on air pollution is the one requirement that is violated by most of the smart services. 
Similarly, \figref{fig:pie} (2) shows the number of violations caused by different smart services. Most of the violations are caused by S1, S5, S6, S7, and S10. The five major services in total cause 71.3\% of the violations. City service developers can also learn from these statistics to adjust the requested actions, the inner logic and parameters of the functions of the services, so that they can design a more compatible service with more acceptable actions in the city.

 \subsubsection{Algorithm Performance}
We compare the average computing time for each requirement under four conditions, (1) using the standard parsing algorithm without the cost function, (2) improved parsing algorithm with a single thread, (3) improved parsing algorithm with spatial parallelization using 4 threads and (4) using 8 threads. The results are shown in \tabref{tab:computingtime}.

\begin{table}[htbp]
\tablefontsize
\caption{Computing time of requirements with standard parsing function, with improved parsing functions and different number of threads}
\label{tab:computingtime}
\begin{tabular}{|l|r| r| r| r|}
\hline
 & \textbf{Standard Parsing (s)} & \textbf{1 thread (s)} & \textbf{4 threads (s)} & \textbf{8 threads (s)} \\ \hline
\textbf{NYR1}         & 2102.13          & 140.29   & 50.31     & \textbf{26.12}     \\ \hline
\textbf{NYR2}         & 55.2             & \textbf{0.837}    & 1.023     & 0.912     \\ \hline
\textbf{NYR3}         & 69.22            & 22.25    & 7.54      & \textbf{4.822}     \\ \hline
\textbf{NYR4}         & 390.19           & 390.19   & 100.23    & \textbf{53.32}     \\ \hline
\textbf{NYR5}         & 61.76            & 61.76    & 20.25     & \textbf{15.68 }    \\ \hline
\textbf{Total}      & 2678.5           & 615.32  & 179.35   & \textbf{100.85}   \\ \hline
\end{tabular}
\vspace{-0.3em}
\end{table}

\begin{figure}[htbp]
    \centering
    \includegraphics[width=8cm]{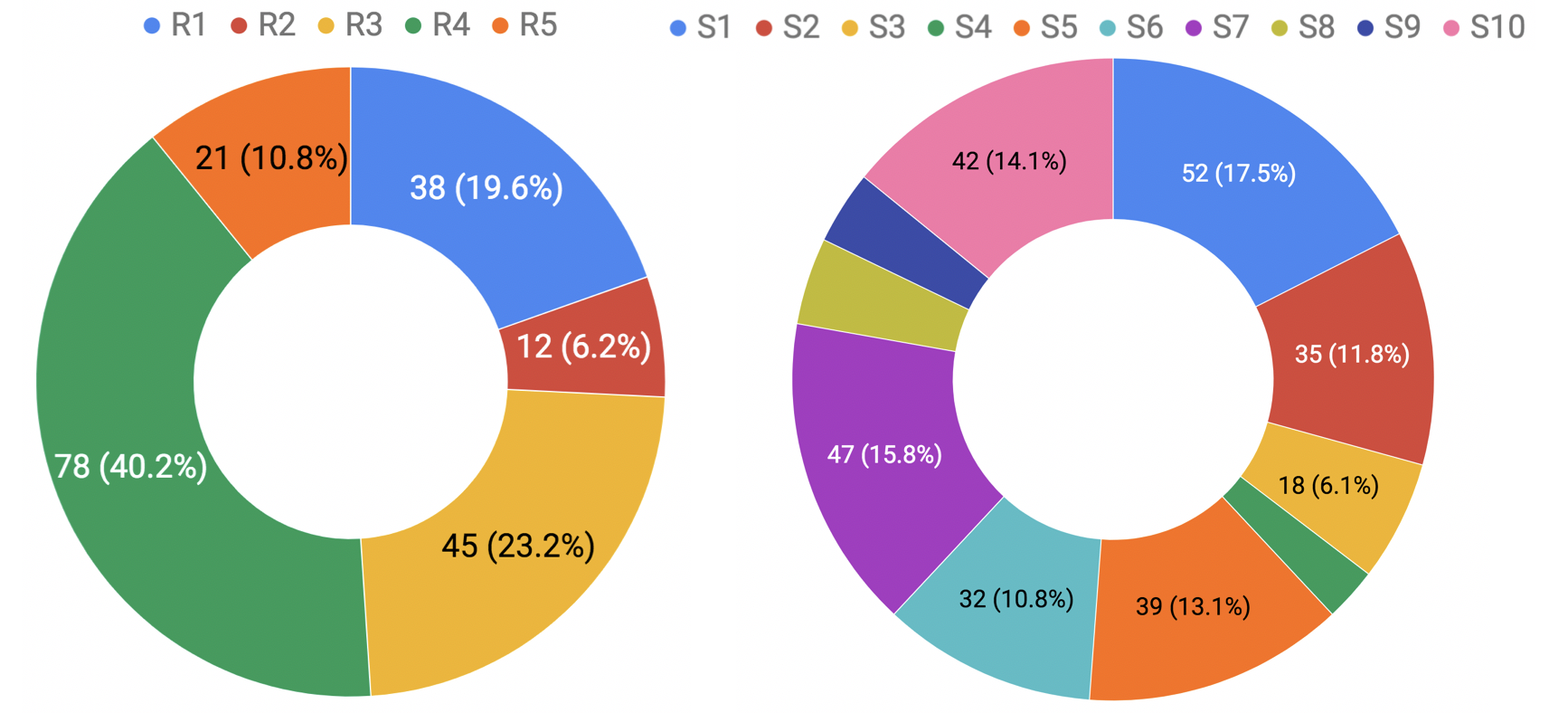}\\
    \scriptsize{(1) Requirements \hspace{1.4cm}(2) Smart Services}
    \caption{Distributions of the violations over requirements and smart services}
    \label{fig:pie}
    \vspace{-0.3em}
\end{figure}

First, the improved parsing algorithm reduces the computing time significantly for the requirement specified on PoIs, especially for NYR1 that computing time reduces from 2102.13 seconds to 140.29 seconds (about 15 times). 

Second, the parallelization over spatial operator further reduces the computing time in most of the cases. For example, for NYR1, the computing time is reduced to 26.12 seconds with 8 threads while 140.29 seconds with single thread (about 5 times). When the amount of data is very small (NYR2), the parallelization time is similar to the single thread time, but still much efficient than the standard parsing.  

The results demonstrate the effectiveness and importance of the efficient monitoring algorithms. In the table, the total time of monitoring 5 requirements is reduced from 2678.5 seconds to 100.85 seconds. In the real world, when multiple requirements are monitored simultaneously, the improvement is extremely important for real-time monitoring.

\subsection{Coverage Analysis}

We compare the specification coverage on 1000 quantitatively-specified real city requirements between STL, SSTL, STREL and SaSTL, the results are shown in \figref{fig:coverage}.  
The study is conducted by graduate students following the rules that if the language is able to specify the whole requirement directly with one single formula, then it is identified as True.
To be noted, another spatial STL, SpaTeL is not considered as a baseline here, because it is not applicable to most of city spatial requirements. SpaTeL is built on a quad tree, and able to specify directions rather than the distance. 

\begin{figure}[htbp]
    \centering
    \includegraphics[width=.5\textwidth]{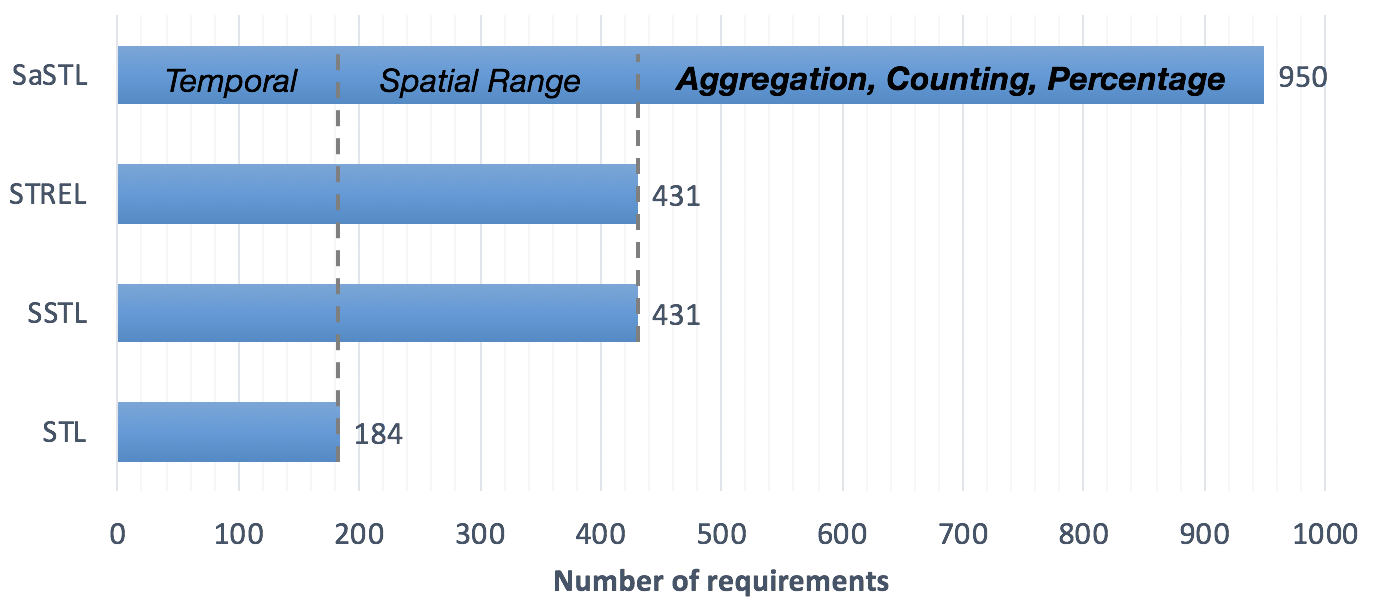}
    \caption{{Comparison of the Specification Coverage on 1000 Real City Requirements}}
    \label{fig:coverage}
    \vspace{-0.8em}
\end{figure}

As shown in \figref{fig:coverage}, STL is only able to specify 184 out of 1000 requirements out, while SSTL and STREL are able to formalize 431 requirements. 
SaSTL is able to specify 950 out of 1000 requirements.
In particular, we made the following observations from the results. 
First, 
50 requirements cannot be specified using any of the four languages because they are defined by complex math formulas that are ambiguous with missing key elements, relevant to the operations of many variables, or referring to a set of other requirements, e.g. ``follow all the requirements from Section 201.12'', etc.  
Secondly, SSTL, STREL and SaSTL outperformed STL in terms of requirements with spatial ranges, such as ``one-mile radius around the entire facility''; Third, SSTL and STREL have the same coverage on the requirements that only contain a temporal and spatial range. Comparing to SSTL and SaSTL, STREL can also be applied to dynamic graph and check requirements reachability, which is very useful in applications like wireless sensor networks, but not common in smart city requirements;
Fourth, for the rest of the requirements (467 out of 1000) that can only be specified by SaSTL contains a various set of locations and the aggregation over spatial ranges.

\section{Related Work}
\label{sect:related}


Smart cities are special cases of cyber-physical systems 
(CPS) featuring a large amount of spatial distributed 
components that are responsible for monitoring and 
controlling the physical environment.
In the last decade, new formal specification languages 
have been proposed to define spatial-temporal requirements 
over the execution of CPS with spatial distributed components.
For example, Talcott in~\cite{Talcott08} has introduced 
the notion of a spatial-temporal event-based model for CPS. 
In such a model, the execution of actions, the exchange 
of messages and the physical changes trigger events that 
are processed by a monitor after being labeled  with 
time and space stamps.
This concept is further elaborated in~\cite{TVG09} 
where a 2D Cartesian coordinate  with location points and 
location fields system is used to represent the space. 
However, the approaches proposed in~\cite{Talcott08,TVG09} 
lack a spatio-temporal logic formalism  
enabling the specification of the requirements and the 
automatic monitoring generation.
Other mathematical structures such as \emph{topological spaces}, 
\emph{closure spaces}, \emph{quasi-discrete closure spaces} and 
\emph{finite graphs}~\cite{NenziBCLM15} can be used to reason 
about spatial relations, such as \emph{closeness} and \emph{neighborhood}
in the context of \emph{collective adaptive systems}~\cite{CianciaLLM16}.
Monitoring spatial distributed cyber-physical systems 
requires to deal with real-value signals over 
continuous time.  Therefore, a common  
approach to design a specification language for 
spatio-temporal properties consists in extending 
{\it STL}~\cite{Maler2004} with spatial operators. 
Recent examples are SpaTeL~\cite{bartocci2015}, 
SSTL~\cite{NenziBCLM15} and STREL~\cite{BartocciBLN17}.




Despite being well-defined and used in several CPS applications, none of these logic systems can be directly applied to smart city requirements specification. The SaSTL proposed in this paper is meant to fill the blank between the formal logics and smart cities. Comparing to the spatial STL discussed above, SaSTL (1) specifies the spatial operators with distance and PoIs from the real world,  (2) adds new aggregation operators, (3) adds a counting operator, and (4) monitors efficiently with parallel processing. 
\section{Conclusion}
\label{sect:sum}


In this paper, we present a novel Spatial Aggregation Signal Temporal Logic (SaSTL) to specify and to monitor real-time safety requirements of smart cities at runtime. 
We develop an efficient monitoring framework 
that optimizes the requirement parsing process and can check in parallel a SaSTL requirement over multiple data streams generated  
from thousands of sensors that are typically spatially distributed over a smart city.

SaSTL is a powerful specification language for smart cities because of its capability to monitor the city desirable features of temporal (e.g. real-time, interval), spatial (e.g., PoIs, range) and their complicated relations (e.g. always, everywhere, aggregation) between them.  
More importantly, it can coalesce many requirements into a single SaSTL formula and provide the aggregated results efficiently, which is a major advance on what smart cities do now. 
We believe it is a valuable step towards developing a practical smart city monitoring system 
even though there are still open issues for future work. Furthermore, SaSTL monitor is more than a smart city monitor, it can also be easily generalized and applied to monitor other large-scale IoT deployments with a large number of sensors and actuators across the spatial domain at runtime efficiently.



\section*{Acknowledgment}
Supported in part by NSF NRT Grant 1829004.


\bibliographystyle{IEEEtran}
\bibliography{monibib}

\begin{thebibliography}{10}
\providecommand{\url}[1]{#1}
\csname url@samestyle\endcsname
\providecommand{\newblock}{\relax}
\providecommand{\bibinfo}[2]{#2}
\providecommand{\BIBentrySTDinterwordspacing}{\spaceskip=0pt\relax}
\providecommand{\BIBentryALTinterwordstretchfactor}{4}
\providecommand{\BIBentryALTinterwordspacing}{\spaceskip=\fontdimen2\font plus
\BIBentryALTinterwordstretchfactor\fontdimen3\font minus
  \fontdimen4\font\relax}
\providecommand{\BIBforeignlanguage}[2]{{%
\expandafter\ifx\csname l@#1\endcsname\relax
\typeout{** WARNING: IEEEtran.bst: No hyphenation pattern has been}%
\typeout{** loaded for the language `#1'. Using the pattern for}%
\typeout{** the default language instead.}%
\else
\language=\csname l@#1\endcsname
\fi
#2}}
\providecommand{\BIBdecl}{\relax}
\BIBdecl

\bibitem{arrayofthings}
C.~E. Catlett, P.~H. Beckman, R.~Sankaran, and K.~K. Galvin, ``Array of things:
  a scientific research instrument in the public way: platform design and early
  lessons learned,'' in \emph{Proceedings of the 2nd International Workshop on
  Science of Smart City Operations and Platforms Engineering}.\hskip 1em plus
  0.5em minus 0.4em\relax ACM, 2017, pp. 26--33.

\bibitem{rio2012center}
\relax New York~Times, ``{IBM} takes `smarter cities’ to rio de janeiro,''
  2012.

\bibitem{cisco-center}
Cisco, ``Smart+connected operations center,'' 2017.

\bibitem{ma2019data}
M.~Ma, S.~M. Preum, M.~Y. Ahmed, W.~T{\"a}rneberg, A.~Hendawi, and J.~A.
  Stankovic, ``Data sets, modeling, and decision making in smart cities: A
  survey,'' \emph{ACM Transactions on Cyber-Physical Systems}, vol.~4, no.~2,
  pp. 1--28, 2019.

\bibitem{yuan2018dynamic}
Y.~Yuan, D.~Zhang, F.~Miao, J.~A. Stankovic, T.~He, G.~Pappas, and S.~Lin,
  ``Dynamic integration of heterogeneous transportation modes under disruptive
  events,'' in \emph{2018 ACM/IEEE 9th International Conference on
  Cyber-Physical Systems (ICCPS)}.\hskip 1em plus 0.5em minus 0.4em\relax IEEE,
  2018, pp. 65--76.

\bibitem{ma2017cityguard}
M.~Ma, S.~M. Preum, and J.~A. Stankovic, ``Cityguard: A watchdog for
  safety-aware conflict detection in smart cities,'' in \emph{Proceedings of
  the Second International Conference on Internet-of-Things Design and
  Implementation}, 2017, pp. 259--270.

\bibitem{zhang2018detecting}
H.~Zhang, Y.~Zheng, and Y.~Yu, ``Detecting urban anomalies using multiple
  spatio-temporal data sources,'' \emph{ACM on Interactive, Mobile, Wearable
  and Ubiquitous Technologies}, vol.~2, no.~1, p.~54, 2018.

\bibitem{sheng2019case}
S.~Sheng, E.~Pakdamanian, K.~Han, B.~Kim, P.~Tiwari, I.~Kim, and L.~Feng, ``A
  case study of trust on autonomous driving,'' in \emph{2019 IEEE Intelligent
  Transportation Systems Conference (ITSC)}.\hskip 1em plus 0.5em minus
  0.4em\relax IEEE, 2019, pp. 4368--4373.

\bibitem{ma2017runtime}
M.~Ma, J.~A. Stankovic, and L.~Feng, ``Runtime monitoring of safety and
  performance requirements in smart cities,'' in \emph{1st ACM Workshop on the
  Internet of Safe Things}, 2017.

\bibitem{haghighi2015spatel}
I.~Haghighi, A.~Jones, Z.~Kong, E.~Bartocci, R.~Gros, and C.~Belta, ``Spatel: a
  novel spatial-temporal logic and its applications to networked systems,'' in
  \emph{Proceedings of the 18th International Conference on Hybrid Systems:
  Computation and Control}.\hskip 1em plus 0.5em minus 0.4em\relax ACM, 2015,
  pp. 189--198.

\bibitem{ma2018cityresolver}
M.~Ma, J.~A. Stankovic, and L.~Feng, ``Cityresolver: a decision support system
  for conflict resolution in smart cities,'' in \emph{Proceedings of the 9th
  ACM/IEEE International Conference on Cyber-Physical Systems}.\hskip 1em plus
  0.5em minus 0.4em\relax IEEE Press, 2018, pp. 55--64.

\bibitem{Maler2004}
O.~Maler and D.~Nickovic, ``Monitoring temporal properties of continuous
  signals,'' in \emph{Proc. FORMATS}, 2004.

\bibitem{NenziBCLM15}
L.~Nenzi, L.~Bortolussi, V.~Ciancia, M.~Loreti, and M.~Massink, ``Qualitative
  and quantitative monitoring of spatio-temporal properties,'' in \emph{Runtime
  Verification - 6th International Conference, {RV} 2015}, vol. 9333.\hskip 1em
  plus 0.5em minus 0.4em\relax Springer, 2015, pp. 21--37.

\bibitem{BartocciBLN17}
E.~Bartocci, L.~Bortolussi, M.~Loreti, and L.~Nenzi, ``Monitoring mobile and
  spatially distributed cyber-physical systems,'' in \emph{{MEMOCODE} 2017: the
  15th {ACM-IEEE} International Conference on Formal Methods and Models for
  System Design}.\hskip 1em plus 0.5em minus 0.4em\relax {ACM}, 2017, pp.
  146--155.

\bibitem{r1}
\BIBentryALTinterwordspacing
NYC.gov, ``Emissions from transportation, nyc environment protection,'' 2019.
  [Online]. Available:
  \url{https://www1.nyc.gov/html/dep/html/air/emissions_from_transportation.shtml}
\BIBentrySTDinterwordspacing

\bibitem{r10}
\relax District~of Columbia Municipal~Regulations and D.~of~Columbia~Register,
  ``Air quality - motor vehicular pollutants, lead, odors, and nuisance
  pollutants,'' 2016.

\bibitem{r2}
S.~Matteo and J.~Brannan, ``A local law to amend the administrative code of the
  city of new york, in relation to restricting the use of bus lanes by
  sight-seeing buses,'' in \emph{Restricting the use of bus lanes by
  sight-seeing buses}.\hskip 1em plus 0.5em minus 0.4em\relax The New York City
  Council, 2019.

\bibitem{r3}
\relax NYC Environment~Protection, ``Use of heating oil remaining in
  tanks.''\hskip 1em plus 0.5em minus 0.4em\relax The city of New York, 2019.

\bibitem{r4}
\relax United States Environmental Protection~Agency, ``Residential energy
  efficiency,'' in \emph{Energy Resources for State and Local
  Governments}.\hskip 1em plus 0.5em minus 0.4em\relax The city of New York,
  2019.

\bibitem{r11}
\relax LA~Sec 111.03. Minimum Ambient Noise~Level, ``Official city of los
  angeles municipal code,'' 2016.

\bibitem{r6}
\relax Hong~Kong, ``Guide to indoor air quality management in hong kong
  regional offices and public places,'' in \emph{Guide to Indoor Air Quality
  Management}, 2019.

\bibitem{r7}
NYC.gov, ``Stopping, standing or parking prohibited in specified places,'' in
  \emph{New York Public Law}, 2016.

\bibitem{r8}
\relax Beijing Emergency~Agency, ``Pre-hospital medical emergency
  regulations,'' 2016.

\bibitem{r9}
\relax Beijing~Government, ``Safety management for kindergarten, primary and
  secondary school,'' 2016.

\bibitem{donze2013efficient}
A.~Donz{\'e}, T.~Ferrere, and O.~Maler, ``Efficient robust monitoring for
  {STL},'' in \emph{International Conference on Computer Aided
  Verification}.\hskip 1em plus 0.5em minus 0.4em\relax Springer, 2013, pp.
  264--279.

\bibitem{Lueker78}
G.~S. Lueker, ``A data structure for orthogonal range queries,'' in \emph{19th
  Annual Symposium on Foundations of Computer Science}.\hskip 1em plus 0.5em
  minus 0.4em\relax {IEEE} Computer Society, 1978, pp. 28--34.

\bibitem{chicagocrime}
\relax City~of Chicago, ``Crimes of {Chicago} - one year prior to present,''
  2018.

\bibitem{ma2016detection}
M.~Ma, S.~M. Preum, W.~Tarneberg, M.~Ahmed, M.~Ruiters, and J.~Stankovic,
  ``Detection of runtime conflicts among services in smart cities,'' in
  \emph{2016 IEEE International Conference on Smart Computing
  (SMARTCOMP)}.\hskip 1em plus 0.5em minus 0.4em\relax IEEE, 2016, pp. 1--10.

\bibitem{sumo}
M.~Behrisch, L.~Bieker, J.~Erdmann, and D.~Krajzewicz, ``Sumo--simulation of
  urban mobility: an overview,'' in \emph{Proceedings of SIMUL 2011}.\hskip 1em
  plus 0.5em minus 0.4em\relax ThinkMind, 2011.

\bibitem{nycopendata}
NYC.gov, \emph{New York City Open Data},
  \url{https://nycopendata.socrata.com/}.

\bibitem{Talcott08}
C.~L. Talcott, ``Cyber-physical systems and events,'' in
  \emph{Software-Intensive Systems and New Computing Paradigms - Challenges and
  Visions}, ser. LNCS.\hskip 1em plus 0.5em minus 0.4em\relax Springer, 2008,
  vol. 5380, pp. 101--115.

\bibitem{TVG09}
Y.~Tan, M.~C. Vuran, and S.~Goddard, ``Spatio-temporal event model for
  cyber-physical systems,'' in \emph{2009 29th IEEE International Conference on
  Distributed Computing Systems Workshops}.\hskip 1em plus 0.5em minus
  0.4em\relax IEEE, 2009, pp. 44--50.

\bibitem{CianciaLLM16}
V.~Ciancia, D.~Latella, M.~Loreti, and M.~Massink, ``Spatial logic and spatial
  model checking for closure spaces,'' in \emph{Proc. of {SFM} 2016}, ser.
  LNCS, vol. 9700.\hskip 1em plus 0.5em minus 0.4em\relax Springer, 2016, pp.
  156--201.

\bibitem{bartocci2015}
I.~Haghighi, A.~Jones, J.~Z. Kong, E.~Bartocci, G.~R., and C.~Belta,
  ``{SpaTeL}: {A} {Novel} {Spatial}-{Temporal} {Logic} and {Its} {Applications}
  to {Networked} {Systems},'' in \emph{Proc. of HSCC}, 2015.

\end{thebibliography}


\end{document}